\documentclass{llncs}

\usepackage{graphicx}

\usepackage{amsmath}

\newcommand{\comment}[1]{}

\newcommand{\fillbox}{\hspace*{\fill}\(\Box\)}

\usepackage{graphicx}
\usepackage{latexsym}

\newcommand{\scripte}{\mathcal{E}}
\newcommand{\scriptv}{\mathcal{V}}
\newcommand{\scriptl}{\mathcal{L}}
\newcommand{\scripts}{\mathcal{S}}
\newcommand{\graphh}{\textit{H}}

\newcommand{\bfA}{{\bf A}}
\newcommand{\bfH}{{\bf H}}
\newcommand{\bfQ}{{\bf Q}}
\newcommand{\bfM}{{\bf M}}
\newcommand{\bfv}{v}

\newcommand{\matrixm}{\textbf{M}}
\newcommand{\matrixh}{\textbf{H}}

\begin{document}

\title{Iterative Approximate Consensus in the presence of Byzantine Link Failures
\thanks{\normalsize This research is supported in part by National Science Foundation award CNS 1329681. Any opinions, findings, and conclusions or recommendations expressed here are those of the authors and do not necessarily reflect the views of the funding agencies or the U.S. government.}}

\author{Lewis Tseng$^{1}$, and Nitin Vaidya$^{2}$}

\institute{ \normalsize $^1$ Department of Computer Science,\\
 \normalsize $^2$ Department of Electrical and Computer Engineering,
 and\\ \normalsize University of Illinois at Urbana-Champaign\\ \normalsize Email: \{ltseng3, nhv\}@illinois.edu\\~\\~\\~Technical Report}


\maketitle

\begin{abstract}{\normalfont
This paper explores the problem of reaching approximate consensus in synchronous point-to-point networks, where each directed link of the underlying communication graph represents a communication channel between a pair of nodes. We adopt the {\em transient Byzantine link} failure model \cite{Santoro_link,Santoro_link2}, where an omniscient adversary controls a subset of the {\em directed} communication links, but the nodes are assumed to be {\em fault-free}.

~

Recent work has addressed the problem of reaching approximate consensus in incomplete graphs with Byzantine {\em nodes} using a {\em restricted class} of iterative algorithms that maintain only a small amount of memory across iterations \cite{vaidya_PODC12,Tseng_general,vaidya_icdcn14,Sundaram_condition}. However, to the best of our knowledge, we are the first to consider approximate consensus in the presence of Byzantine {\em links}. We extend our past work that provided exact characterization of graphs in which the iterative approximate consensus problem in the presence of Byzantine {\em node} failures is solvable \cite{vaidya_PODC12,Tseng_general}. In particular, we prove a {\em tight} necessary and sufficient condition on the underlying communication graph for the existence of iterative approximate consensus algorithms under {\em transient Byzantine link} model. The condition answers (part of) the open problem stated in \cite{Santoro_link2}.
}
\end{abstract}




\section{Introduction}
\label{s_intro}
Approximate consensus can be related to many distributed computations in networked systems, such as data aggregation \cite{Kempe_gossip}, decentralized estimation  \cite{noisy_link}, and flocking \cite{Jadbabaie}. Extensive work has addressed the problem in the presence of {\em Byzantine nodes} \cite{psl_BG_1982} in either complete networks \cite{AA_Dolev_1986,AA_optimal} or arbitrary directed networks \cite{vaidya_PODC12,Sundaram_condition,Tseng_general}. As observed in \cite{Biely_hybrid,Schmid_link}, link failures become more and more prevalent. Thus, it is of interest to consider the problem of approximate consensus in the presence of Byzantine {\em link} failures.

This paper explores such problem in synchronous point-to-point networks, where each directed link of the underlying communication graph represents a communication channel between a pair of nodes. The link failures are modeled using a {\em transient Byzantine link} failure model (formal definition in Section \ref{s_models}) \cite{Santoro_link,Santoro_link2}, in which different sets of link failures may occur at different time. We consider the problem in arbitrary directed graphs using a {\em restricted class} of iterative algorithms that maintain only a small amount of memory across iterations, e.g., the algorithms do not require the knowledge of the network topology. Such iterative algorithms are of interest in networked systems, since they have low complexity and do not rely on global knowledge \cite{Sundaram_condition}. In particular, the iterative algorithms have the following properties, which we will state more formally later:

\begin{itemize}
\item {\bf Initial state} of each node is equal to a real-valued {\em input} provided to that node.

\item {\bf Termination}: The algorithm terminates in finite number of iterations.

\item {\bf Validity}: After each iteration of the algorithm, the state of each node must stay in the {\em convex hull} of the states of all the nodes at the end of the {\em previous} iteration.

\item {\bf $\epsilon$-agreement}: For any $\epsilon > 0$, when the algorithm terminates, the difference between any pair of nodes is guaranteed to be within $\epsilon$.
\end{itemize}

\paragraph{Main Contribution}

This paper extends our recent work on approximate consensus under node failures \cite{vaidya_PODC12,Tseng_general}. The main contribution is identifying a {\em tight} necessary and sufficient condition for the graphs to be able to reach approximate consensus under {\em transient Byzantine link} failure models \cite{Santoro_link,Santoro_link2} using restricted iterative algorithms; our proof of correctness follows a structure previously used in our work to prove correctness of other consensus algorithms in incomplete networks \cite{Tseng_general,vaidya_icdcn14}. The use of matrix analysis is inspired by the prior work on non-fault-tolerant consensus (e.g., \cite{Jadbabaie,AA_convergence_markov}).

\paragraph{Related Work}

Approximate consensus has been studied extensively in synchronous as well as asynchronous systems. Bertsekas and Tsitsiklis explored reaching approximate consensus without failures in synchronous dynamic network, where the underlying communication graph is time-varying \cite{AA_convergence_markov}. Dolev et al. considered approximate consensus in the presence of {\em Byzantine nodes} in both synchronous and asynchronous systems \cite{AA_Dolev_1986}, where the network is assumed to be a clique, i.e., a complete network. Subsequently, for complete graphs, Abraham et al. proposed an algorithm to achieve approximate consensus with {\em Byzantine nodes} in asynchronous systems using optimal number of nodes \cite{AA_optimal}. 

Recent work has addressed approximate consensus in incomplete graphs with faulty {\em nodes} \cite{vaidya_PODC12,Sundaram_condition,Tseng_general}. \cite{vaidya_PODC12,Tseng_general} and \cite{Sundaram_condition} showed exact characterizations of graphs in which the approximate consensus problem is solvable in the presence of Byzantine nodes and malicious nodes, respectively. Malicious node is a restricted type of Byzantine node in which every node is forced to send the identical message to all of its neighbors. 

Much effort has also been devoted to the problem of achieving consensus in the presence of link failures \cite{HeardOf,Biely_hybrid,Schmid_link,Santoro_link,Santoro_link2}. Charron-Bost and Schiper proposed a HO (Heard-Of) model that captures both the link and node failures at the same time \cite{HeardOf}. However, the failures are assumed to be benign in the sense that no corrupted message will ever be received in the network. Santoro and Widmayer proposed the {\em transient} Byzantine link failure model: a different set of links can be faulty at different time \cite{Santoro_link,Santoro_link2}. They characterized a necessary condition and a sufficient condition for undirected networks to achieve consensus in the transient link failure model; however, the conditions are {\em not} tight (i.e., do not match): necessary and sufficient conditions are specified in terms of node degree and edge-connectivity,\footnote{A graph $G=(\scriptv, \scripte)$ is said to be $k$-edge connected, if $G'=(\scriptv, \scripte - X)$ is connected for all $X \subseteq \scripte$ such that $|X| < k$.} respectively. Subsequently, Biely et al. proposed another link failure model that imposes an upper bound on the number of faulty links incident to each node \cite{Biely_hybrid}. As a result, it is possible to tolerate $O(n^2)$ link failures with $n$ nodes in the new model. Under this model, Schmid et al. proved lower bounds on number of nodes, and number of rounds for achieving consensus \cite{Schmid_link}. However, incomplete graphs were not considered in \cite{Biely_hybrid,Schmid_link}.

For consensus problem, it has been shown in \cite{impossible} and \cite{Santoro_link2}, respectively, that an undirected graph of $2f+1$ node-connectivity\footnote{A graph $G=(\scriptv, \scripte)$ is said to be $k$-node connected, if $G'=(\scriptv - X, \scripte)$ is connected for all $X \subseteq \scriptv$ such that $|X| < k$.} and edge-connectivity is able to tolerate $f$ Byzantine nodes and $f$ Byzantine links. Independently, researchers showed that $2f+1$ node-connectivity is both necessary and sufficient for the problem of information dissemination in the presence of either $f$ faulty nodes \cite{SS_node} or $f$ {\em fixed} faulty links \cite{SS_link}.\footnote{Unlike the ``transient" failures in our model, the faulty links are assumed to be fixed throughout the execution of the algorithm in \cite{SS_link}.} However, both node-connectivity and edge-connectivity are not adequate for our problem as illustrated in Section \ref{sec:iacbl}. 

Link failures have also been addressed under other contexts, such as distributed method for wireless control network \cite{control}, reliable transmission over packet network \cite{packet}, or estimation over noisy links \cite{noisy_link}.

\section{System Model}
\label{s_models}

{\em Communication model:}
The system is assumed to be {\em synchronous}. 
The communication network is modeled as a simple {\em directed} graph $G(\scriptv,\scripte)$, where $\scriptv=\{1,\dots,n\}$ is the set of $n$ nodes, and $\scripte$ is the set of directed edges between the nodes in $\scriptv$.
With a slight abuse of terminology, we will use the terms {\em edge}
and {\em link} interchangeably in our presentation.
In simple graph, there is at most one directed edge from any node $i$ to some other node $j$ (But our results can be extended to multi-graph).
We assume that $n\geq 2$, since the consensus problem for $n=1$ is trivial.
Node $i$ can reliably transmit messages to node $j$ if and only if
the directed edge $(i,j)$ is in $\scripte$.
Each node can send messages to itself as well; however,
for convenience, we exclude {\em self-loops} from set $\scripte$.
That is, $(i,i)\not\in\scripte$ for $i\in\scriptv$.

For each node $i$, let $N_i^-$ be the set of nodes from which $i$ has incoming edges.
That is, $N_i^- = \{\, j ~|~ (j,i)\in \scripte\, \}$.
Similarly, define $N_i^+$ as the set of nodes to which node $i$
has outgoing edges. That is, $N_i^+ = \{\, j ~|~ (i,j)\in \scripte\, \}$.
Since we exclude self-loops from $\scripte$,
$i\not\in N_i^-$ and $i\not\in N_i^+$. 
However, we note again that each node can indeed send messages to itself. Similarly, let $E_i^-$ be the set of incoming links incident to node $i$. That is, $E_i^-$ contains all the links from nodes in $N_i^-$ to node $i$, i.e., $E_i^- = \{(j,i)~|~j \in N_i^-\}$.\\

{\em Failure Model:}
We consider the transient Byzantine {\em link} failure model \cite{Santoro_link,Santoro_link2} for iterative algorithms in directed network. All nodes are assumed to be {\em fault-free}, and only send a single message once in each iteration. 
A link $(i,j)$ is said to be faulty if the message sent by node $i$ is different from the message received by node $j$ in some iteration. Note that in our model, it is possible that link $(i,j)$ is faulty while link $(j,i)$ is fault-free.\footnote{For example, the described case is possible in wireless network, if node $i$'s transmitter is broken while node $i$'s receiver and node $j$'s transmitter and receiver all function correctly.} In every iteration, up to $f$ links may be faulty, at most $f$ links may deliver incorrect message or drop message. Note that different sets of link failures may occur in different iterations.

A faulty link may tamper or drop messages. Also, the faulty links may be controlled by a single omniscient adversary. That is, the adversary is assumed to have a complete knowledge of the execution of
the algorithm, including the states of all the nodes,
contents of messages the other nodes send to each other,
the algorithm specification, and the network topology.

\section{IABC Algorithms and Example Network}
\label{sec:iacbl}

In this section, we describe the structure of the
{\em Iterative Approximate Byzantine Consensus} (IABC) algorithms of interest, and state conditions that they must satisfy. The IABC structure is identical to the one in our prior work on node failures \cite{vaidya_PODC12,Tseng_general,vaidya_icdcn14}.

Each node $i$ maintains state $v_i$, with $v_i[t]$ denoting the state
of node $i$ at the {\em end}\, of the $t$-th iteration of the algorithm ($t \geq 0$).
Initial state of node $i$, $v_i[0]$, is equal to the initial {\em input}\, provided to node $i$. At the {\em start} of the $t$-th iteration ($t>0$), the state of
node $i$ is $v_i[t-1]$. We assume that the input at each node is lower bounded by a constant $\mu$ and upper bounded by a constant $U$. The
iterative algorithm may terminate after a number of iterations that is a function of $μ$ and $U$. $\mu$ and $U$ are assumed to be known a priori. 

The IABC algorithms of interest will require each node $i$
to perform the following three steps in iteration $t$, where $t>0$.
Note that the message sent via faulty links may deviate from this specification.

\begin{enumerate}
\item {\em Transmit step:} Transmit current state, namely $v_i[t-1]$, on all outgoing edges
 (to nodes in $N_i^+$).

\item {\em Receive step:} Receive values on all incoming edges (from nodes in $N_i^-$). 
Denote by $r_i[t]$ the vector of values received by node $i$ from its
neighbors. The size of vector $r_i[t]$ is $|N_i^-|$. The values sent in iteration $t$ are received in the same iteration (unless dropped by the faulty links).

\item {\em Update step:} Node $i$ updates its state using a transition function $T_i$ as
follows. $T_i$ is a part of the specification of the algorithm, and takes
as input the vector $r_i[t]$ and state $v_i[t-1]$.
\begin{eqnarray}
v_i[t] & = &  T_i ~( ~r_i[t]\,,\,v_i[t-1] ~)
\label{eq:T_i}
\end{eqnarray}

\end{enumerate}

~

The following properties must be satisfied by an IABC algorithm
in the presence of up to $f$ Byzantine faulty links:
\begin{itemize}
\item {\bf Termination}: the algorithm terminates in finite number of iterations.\\

\item {\bf Validity:} $\forall t>0,
~~\min_{i \in \scriptv} v_i[t] \ge \min_{i \in \scriptv} v_i[t-1]$
 and \\
$~~~~~~~~~~~~~~~~~~~~\max_{i \in \scriptv} v_i[t] \ge \max_{i \in \scriptv} v_i[t-1]$.\\

\item {\bf $\epsilon$-agreement:} If the algorithm terminates after $t_{end}$ iterations, then $\forall i, j \in \scriptv, |v_i[t_{end}] - v_j[t_{end}]| < \epsilon$.
\end{itemize}
The objective in this paper is to identify the necessary and sufficient
conditions for the existence of a {\em correct} IABC algorithm (i.e.,
an algorithm satisfying the above properties) for a given $G(\scriptv,\scripte)$.


\paragraph{Example Network}

We give an example showing that node- and edge-connectivity are not adequate for specifying the {\em tight} condition in directed graphs. Consider the case when $f=1$ in the network in Figure \ref{f:eg}. In the network, nodes $A, B, C, D$ form a clique, while node $E$ has only incoming edges from nodes $B, C, D$. It is obvious that the node- and edge-connectivity of the network are less than $2f+1 = 3$, since node $E$ does not have any outgoing links to any other node. However, the approximate consensus is solvable using IABC algorithms under one (directed) faulty link, since the network satisfies the sufficient condition proved later. The proof is presented in \ref{a:example}. Therefore, $2f+1$ node- and edge-connectivity are not necessary for the existence of IABC algorithms.

\begin{figure}[hbt!]
\centering
\includegraphics[width=5cm]{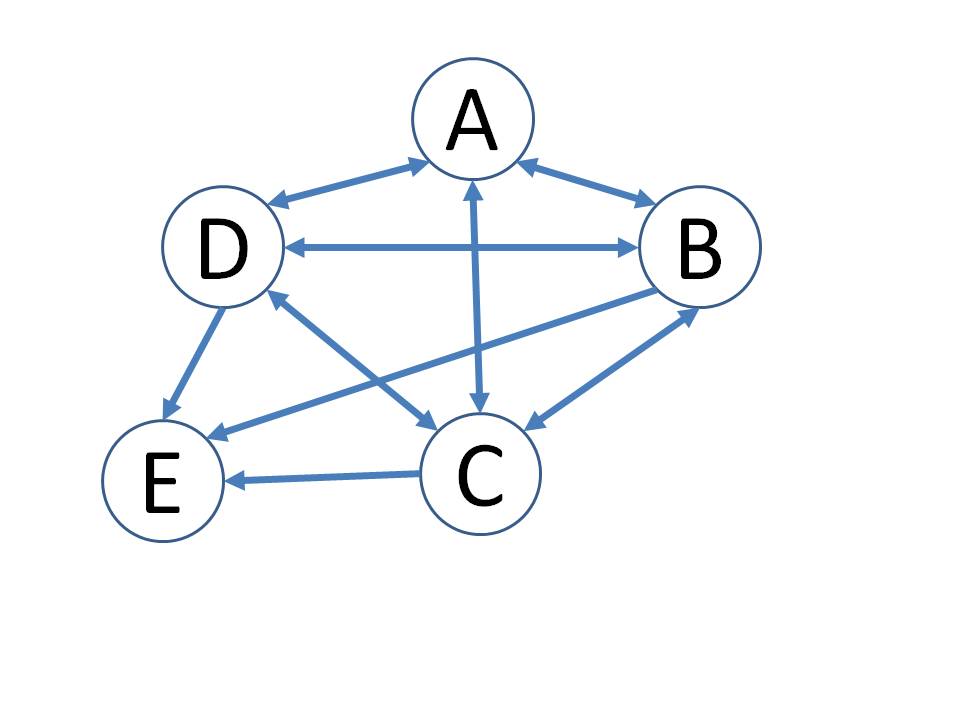}
\vspace*{-10pt}
\caption{Example Network}
\label{f:eg}
\end{figure}

\section{Necessary Condition}
\label{s_necessary}



For a correct iterative approximate consensus algorithm to exists in the presence of Byzantine link failures, the graph $G(\scriptv, \scripte)$ must satisfy the necessary condition proved in this section.
We now define relations $\Rightarrow$
and $\not\Rightarrow$ that are used frequently in our proofs.

\begin{definition}
\label{def:absorb}
For non-empty disjoint sets of nodes $A$ and $B$ in $G(\scriptv, \scripte)$, $A \Rightarrow B$ iff there exists a node $i\in B$ that has at least $f+1$ incoming links from nodes in $A$, i.e., $|\{(j,i)~|~j\in A,~(j,i)\in\scripte\}|>f$; $A\not\Rightarrow B$ iff $A\Rightarrow B$ is {\em not} true.



\end{definition}

~

\noindent {\em Condition P}~: Consider graph $G(\scriptv, \scripte)$. Denote by $F$ a subset of $\scripte$ such that $|F| \leq f$.
Let sets $L,C,R$ form a partition of $\scriptv$, such that
both $L$ and $R$ are non-empty. Then, in $G' = (\scriptv, \scripte-F)$, at least one of the two conditions below must be true: (i) $C\cup R\Rightarrow L$; (ii) $L\cup C\Rightarrow R$.

\begin{theorem}
\label{thm:nc}
Suppose that a correct IABC algorithm exists for $G(\scriptv, \scripte)$. Then $G$ satisfies {\em Condition P}.
\end{theorem}
\begin{proof}

The proof is by contradiction.
Let us assume that a correct IABC algorithm exists,
and for some node partition $L, C, R$ and a subset $F \subseteq \scripte$ such that $|F| \leq f$, $C\cup R\not\Rightarrow L$ and $L\cup C\not\Rightarrow R$ in $G'=(\scriptv, \scripte-F)$.
Thus, for any $i\in L$, $|\{(k,i)~|~k\in C\cup R,~(k,i) \in \scripte-F\}|<f+1$. Similarly,
for any  $j\in R$, $|\{(k,j)~|~k\in L\cup C,~(k,j) \in \scripte-F\}|<f+1$.

Also assume that the links in $F$ (if $F$ is non-empty) all behave faulty, and the rest of the links are all fault-free in every iteration. Note that the nodes are not aware of the identity of the faulty links.

Consider the case when (i) each node in $L$ has initial input $m$, (ii) each
node in $R$ has initial input $M$, such that $M>m$,
and (iii) each node in $C$, if $C$ is non-empty,
has an input in the interval $[m,M]$. Define $m^-$ and $M^+$ such that $m^- <m<M<M^+$.

In the {\em Transmit Step} of iteration 1, each node $k$, sends to nodes in $N_k^+$ value $v_k[0]$; however, some values sent via faulty links may be tampered. Suppose that the faulty links in $F$ (if non-empty) tamper the messages sent via them in the following way (i) if the link is an incoming link to a node in $L$, then $m^- < m$ is deliver to that node;
(ii) if the link is an incoming link to a node in $R$, then $M^+ > M$ is deliver to that node; and (iii) if the link is an incoming link to a node in $C$, then some arbitrary value in interval $[m,M]$ is deliver to that node. This behavior is possible since links in $F$ are Byzantine faulty by assumption. Note that $m^-<m<M<M^+$. 

Consider any node $i \in L$. Recall that $E_i^-$ the set of all the node $i$'s incoming links. Let $E'_i$ be the subset of $E_i^-$ that are incident to nodes in $C\cup R$, i.e., 

\[
E'_i = \{(j, i)~|~j \in C\cup R, (j,i)\in\scripte\}.
\]

Since $|F|\leq f$, $|E_i^-\cap F|\leq f$.
Moreover, by assumption $C\cup R\not\Rightarrow L$; thus, $|E'_i - F| \leq |E'_i|\le f$.
Node $i$ will then receive $m^-$ via the links in $E_i^-\cap F$ (if non-empty)
and values in $[m,M]$ via the links in $E'_i - F$, and
$m$ via the rest of the links, i.e., links in $E_i^- - E_i' - F$.

Consider the following two cases:
\begin{itemize}

\item Both $E_i^-\cap F$ and $E'_i - F$ are non-empty:

In this case, recall that $|E_i^-\cap F|\leq f$ and $|E'_i - F|\leq f$.
From node $i$'s perspective, consider two possible scenarios:
(a) links in $E_i^-\cap F$ are faulty, and the other
links are fault-free, and (b) links in $E'_i - F$ are faulty, and the
other links are fault-free.

In scenario (a), from node $i$'s perspective, all the nodes may have sent values
in interval $[m,M]$, but the faulty links have delivered $m^-$ to node $i$. According to the validity
property, $v_i[1] \geq m$. On the other hand, in scenario (b), all the
nodes may have sent values $m^-$ or $m$, where $m^-<m$; so $v_i[1] \leq m$, according to
the validity property. Since node $i$ does not know whether the
correct scenario is (a) or (b), it must update its state to satisfy the
validity property in both cases. Thus, it follows that $v_i[1] = m$.

\item
At most one of $E_i^-\cap F$ and $E'_i - F$ is non-empty: 

Recall that by assumption, $|E_i^-\cap F|\leq f$ and $|E'_i - F|\leq f$. Since at most one of the set is non-empty, 
$|(E_i^-\cap F)\cup (E'_i-F)|\leq f $.
From node $i$'s perspective,
it is possible that the links in $(E_i^-\cap F)\cup (E'_i - F)$ are all faulty,
and the rest of the links are fault-free. 
In this situation, the values sent to node $i$ via all the fault-free links are all $m$, and therefore, $v_i[1]$ must be set to $m$
as per the validity property.

\end{itemize}
Thus, $v_i[1]=m$ for each node $i\in L$.
Similarly, we can show that $v_j[1] = M$ for each node $j \in R$.

Now consider the nodes in set $C$, if $C$ is non-empty.
All the values received by the nodes in $C$ are in $[m,M]$, therefore,
their new state must also remain in $[m,M]$, as per the {\em validity} property.

The above discussion implies that, at the end of iteration 1,
the following conditions hold true: (i) state of each node in $L$ is
$m$, (ii) state of each node in $R$ is $M$, and (iii) state of each node
in $C$ is in the interval $[m,M]$. These conditions are identical to the initial conditions
listed previously. Then, by a repeated application of the above
argument (proof by induction), it follows that for
any $t \geq 0$, $v_i[t] = m$ for all $\forall i \in L$, $v_j[t] = M$
for all $j \in R$ and $v_k[t]\in[m,M]$ for all $k\in C$.

Since both $L$ and $R$ are non-empty, the {\em $\epsilon$-agreement} property
is not satisfied. A contradiction.
\fillbox
\end{proof}

Theorem \ref{thm:nc} shows that {\em Condition P} is necessary. However, {\em Condition P} is not intuitive. Below, we state an equivalent condition {\em Condition S} that is
easier to interpret. To facilitate the
statement, we introduce the notions of ``source component'' and ``link-reduced graph'' using the following three definitions. The link-reduced graph is analogous to the similar concept introduced in our prior work on node failures \cite{vaidya_PODC12,Tseng_general,vaidya_icdcn14}. 

\begin{definition}
\label{def:decompose}
{\bf Graph decomposition:}
Let $H$ be a directed graph. Partition graph $H$ into non-empty strongly connected components,
$H_1,H_2,\cdots,H_h$, where $h$ is a non-zero integer dependent on graph $H$,
such that
\begin{itemize}
\item every pair of nodes {within} the same strongly connected component has directed
paths in $H$ to each other, and
\item for each pair of nodes, say $i$ and $j$, that belong to
two {\em different} strongly connected components, either $i$ does not have a
directed path to $j$ in $H$, or $j$ does not have a directed path to $i$ in $H$.
\end{itemize}
Construct a graph $H^d$ wherein each strongly connected component $H_k$ above is represented
by vertex $c_k$, and there is an edge from vertex $c_k$ to vertex $c_l$ if and only if
the nodes in $H_k$ have directed paths in $H$ to the nodes in $H_l$.
\end{definition}
It is known that the decomposition
graph $H^d$ is a directed {\em acyclic} graph \cite{dag_decomposition}.

\begin{definition}
{\bf Source component}:
Let $H$ be a directed graph, and let $H^d$ be its decomposition as per
Definition~\ref{def:decompose}. 
Strongly connected component $H_k$ of $H$ is said to be a {\em source component}
if the corresponding vertex $c_k$ in $H^d$ is \underline{not} reachable from any
other vertex in $H^d$. 
\end{definition}

\begin{definition}
\label{def:reduced} {\bf Link-Reduced Graph:}
For a given graph $G(\scriptv,\scripte)$ and $F\subset\scripte$,
a graph $G_F(\scriptv,\scripte_F)$
is said to be a {\em link-reduced graph}, if
$\scripte_F$ is obtained by first removing from $\scripte$ all the links in $F$, and {\em then} removing up to $f$ other incoming
links at each node in $\scripte - F$.
\end{definition}
Note that for a given $G(\scriptv,\scripte)$ and a given $F$,
multiple link-reduced graphs $G_F$ may exist. \\

Now, we state {\em Condition S}:

~

\noindent {\em Condition S}: Consider graph $G(\scriptv,\scripte)$. For any $F \subseteq \scripte$ such that $|F| \leq f$, every
link-reduced graph $G_F$ obtained as per Definition \ref{def:reduced} 
must contain exactly one {\em source component}.

Then, we show that {\em Condition S} and {\em Condition P} specify the equivalent property of the graph.

\begin{lemma}
\label{lemma:p-to-s}
Suppose that {\em Condition P} holds for graph $G(\scriptv,\scripte)$. Then $G$ satisfies {\em Condition S}.

\end{lemma}

\begin{proof}
By assumption, 
$G$ contains at least two node, and so does $G_F$;
therefore, at least one
source component must exist in $G_F$. We now prove that $G_F$ cannot
contain more than one source component. The proof is by contradiction.
Suppose that there exists a subset $F\subset \scripte$ with $|F|\leq f$,
and the link-reduced graph 
$G_F(\scriptv,\scripte_F)$ corresponding to $F$ such
that the decomposition of $G_F$ includes at least two source components.

Let the sets of nodes in two such source components of $G_F$
be denoted $L$ and $R$, respectively. Let $C=\scriptv-L-R$.
Observe that $L,C,R$ form a partition of the nodes in $\scriptv$.
Since $L$ is a source component in $G_F$, it follows that
there are no directed links in $\scripte_F$ from any node in
$C\cup R$ to the nodes in $L$.
Similarly, since $R$ is a source component in $G_F$, it follows that
there are no directed links in $\scripte_F$ from any node in $L\cup C$ to
the nodes in $R$.
These observations, together with the manner in which $\scripte_F$
is defined, imply that (i) there are at most $f$ links in $\scripte - F$ from
the nodes in $C\cup R$ to each node in $L$, and
(ii) there are at most $f$ links in $\scripte - F$ from
the nodes in $L\cup C$ to each node in $R$.
Therefore, in graph $G' = (\scriptv,\scripte-F)$, $C\cup R\not\Rightarrow L$ and $L\cup C\not\Rightarrow R$. Thus, $G = (\scriptv, \scripte)$ does not satisfies {\em Condition P}, since $F \subseteq \scripte$ and $|F| \leq f$, a contradiction.
\fillbox
\end{proof}

\begin{lemma}
\label{lemma:s-to-p}
Suppose that {\em Condition S} holds for graph $G(\scriptv,\scripte)$.
Then, $G$ satisfies {\em Condition P}.
\end{lemma}

\begin{proof}
The proof is by contradiction. Suppose that {\em Condition P} does not hold for graph $G = (\scriptv,\scripte)$. Thus, there exist a subset $F \subset \scripte$, where $|F| \leq f$, and a node partition $L,C,R$, where $L$ and $R$ are both non-empty, such that $C\cup R\not\Rightarrow L$ and $L\cup C\not\Rightarrow R$ in $G' = (\scriptv, \scripte-F)$.

We now constructed a link-reduced graph $G_F(\scriptv, \scripte_F)$ corresponding to set $F$. First, remove all links in $F$ from $\scripte$. Then since $C\cup R\not\Rightarrow L$, the number of links at each node in $L$ from nodes in $C \cup R$ is at most $f$; remove all these links. Similarly, for every node $j \in R$, remove all links from nodes in $L \cup C$ to $j$ (recall that by assumption, there are at most $f$ such links). The remaining links form the set $\scripte_F$. It should be obvious that $G_F(\scriptv, \scripte_F)$ satisfies Definition \ref{def:reduced}; hence, $G_F$ is a valid link-reduced graph.

Now, observe that by construction, in the link-reduced graph $G_F(\scriptv, \scripte_F)$ , there are no incoming links to nodes in $R$ from nodes in $L\cup C$; similarly, in $G_F$, there are no incoming links to nodes in $L$ from nodes in $C\cup R$. It follows that for each $i \in L$, there is no path using links in $\scripte_F$ from $i$ to nodes in $R$; similarly, for each $j \in R$, there is no path using links in $\scripte_F$ from $j$ to nodes in $L$. Thus, $G_F$ must contain at least two source components. Therefore, the existence of $G_F$ implies that $G$ violates {\em Condition S}, a contradiction.
\fillbox
\end{proof}

~

Lemmas \ref{lemma:p-to-s} and \ref{lemma:s-to-p} imply that {\em Condition P} is equivalent to {\em Condition S}. An alternate interpretation of {\em Condition S} is that in every link-reduced graph $G_F$, non-fault-tolerant iterative consensus must be possible.

\subsection{Useful Properties}
\label{s:useful}

Suppose $G(\scriptv, \scripte)$ satisfies {\em Condition P} and {\em Condition S}. We provide two lemmas below to state some properties of $G(\scriptv, \scripte)$ that are useful for analyzing the iterative algorithm presented later. Lemma \ref{lemma:path} intuitively states that at least one node can propagate its value to all the other nodes (over enough number of iterations). Lemma \ref{lemma:2f+1} states that each node needs to have enough incoming neighbors for achieving approximate consensus. 

\begin{lemma}
\label{lemma:path}
Suppose that graph $G(\scriptv, \scripte)$ satisfies {\em Condition S}. Then, in any link-reduced graph $G_F(\scriptv, \scripte_F)$, there exists a node that has a directed path to all the other nodes in $\scriptv$.
\end{lemma}

\begin{proof}
Recall that {\em Condition S} states that any link-reduced graph $G_F(\scriptv, \scripte_F)$ has a single source component. By the definition of source component, any node in the source component (say node $s$) has directed paths using edges in $\scripte_F$ to all the other nodes in the source component, since the source component is a strongly connected component. Also, by the uniqueness of the source component, all other strongly connected components in $G_F$ (if any exist) are not source components, and hence reachable from the source component using the edges in $\scripte_F$. Therefore, node $s$ also has directed paths to all the nodes in $\scriptv$ that are not in the source component as well. Therefore, node $s$ has directed paths to all the other nodes in $\scriptv$. This proves the lemma.
\fillbox
\end{proof}

\begin{lemma}
\label{lemma:2f+1}
For $f > 0$, if graph $G=(\scriptv, \scripte)$ satisfies {\em Condition P}, then each node in $\scriptv$ has in-degree at least $2f+1$, i.e., for each $i \in \scriptv, |N_i^-| \geq 2f+1$.
\end{lemma}

\begin{proof}
The proof is by contradiction. By assumption in the lemma, $f > 0$, and graph $G = (\scriptv, \scripte)$ satisfies {\em Condition P}.

Suppose that there exists a node $i \in \scriptv$ such that $|N_i^-| \leq 2f$. Define $L = \{i\}, C = \emptyset$, and $R = \scriptv - \{i\}$. Note that sets $L,C,R$ form a partition of $\scriptv$. Now, define an edge set $F$ such that $F \subseteq \scripte, |F| \leq f$, and $F$ contains $\min(f, |N_i^-|)$ incoming links from nodes in $R$ to node $i$.

Observe that $f > 0$, and $|L \cup C| = 1$. Thus,  there can be at most $1$ link from $L \cup C$ to any node in $R$ in $G' = (\scriptv,\scripte-F)$. Therefore, $L\cup C \not\Rightarrow R$ in $G' = (\scriptv,\scripte-F)$. 
Then, recall that $E_i^-$ is the set of all the node $i$'s incoming links. Since $L = \{i\}$ and $C = \emptyset$, $E_i^- = \{(j,i)~|~j \in R\}$. Also, since $|E_i^-| = |N_i^-| \leq 2f$, and $F$ contains $\min(f, |N_i^-|)$ links in $E_i^-$, $|E_i^- - F| \leq 2f - f = f$. Therefore, $C\cup R \not\Rightarrow L$ in $G(\scriptv,\scripte-F)$.
Thus, $G' = (\scriptv,\scripte)$ does not satisfy {\em Condition P}, a contradiction. 
\fillbox
\end{proof}

\section{Algorithm 1}
\label{s:algorithm1}

We will prove that there exists a correct IABC algorithm – particularly Algorithm 1 below – that satisfies the termination, validity and
$\epsilon$-agreement properties provided that the graph $G(\scriptv, \scripte)$ satisfies {\em Condition S}. This implies that
{\em Condition P} and {\em Condition S} ares also sufficient. Algorithm 1 has the iterative structure described in Section \ref{sec:iacbl}, and
it is similar to algorithms that were analyzed in prior work as
well \cite{vaidya_PODC12,Tseng_general} (although correctness of the algorithm under
the necessary condition ({\em Conditions P} and {\em S}) has not been proved previously).

~

\hrule
{\bf Algorithm 1}
\vspace*{4pt}\hrule

\begin{enumerate}

\item {\em Transmit step:} Transmit current state $v_i[t-1]$ on all outgoing edges and self-loop.

\item {\em Receive step:} Receive values on all incoming edges and self-loop. These values form vector $r_i[t]$ of size $|N_i^-|+1$ (including the value from node $i$ itself). When a node expects to receive a message from an incoming neighbor but does not receive the message, the message value is assumed to be equal to its own state, i.e., $v_i[t-1]$.

\item {\em Update step:}
Sort the values in $r_i[t]$ in an increasing order (breaking ties arbitrarily), and eliminate the smallest and largest $f$ values. 
Let $N_i^*[t]$ denote the set of nodes from whom the remaining $|N_i^-| +1 - 2f$ values in $r_i[t]$ were received. Note that as proved in Lemma \ref{lemma:2f+1}, each node has at least $2f+1$ incoming neighbors. Thus, when $f > 0$, $|N_i^*[t]| \geq 2$.
Let $w_j$ denote the value received from node $j \in N_i^*[t]$. Note that $i \in N_i^*[t]$. Hence, for convenience, define $w_i=v_i[t-1]$ to be the value node $i$ receives from itself.  Observe that
if the link from $j\in N_i^*[t]$ is fault-free, then $w_j=v_j[t-1]$.

Define
\begin{eqnarray}
v_i[t] ~ = ~ T_i(r_i[t]) ~ = ~\sum_{j\in N_i^*[t]} a_i \, w_j
\label{e_T}
\end{eqnarray}
where
\[ a_i = \frac{1}{|N_i^*[t]|} = \frac{1}{|N_i^-|+1-2f}
\] 

The ``weight'' of each term on the right-hand side of
(\ref{e_T}) is $a_i$. Note that $|N_i^*[t]| = |N_i^-|+1 - 2f$, and $i\not\in N_i^*[t]$ because $(i,i)\not\in\scripte$. Thus, the weights on the right-hand side add to 1. Also, $0<a_i\leq 1$.\footnote{Although $f$ and $a_i$ may be different for each iteration $t$, for simplicity, we do not explicitly represent this dependence on $t$ in the notations.} 


\end{enumerate}

\noindent {\bf Termination}: Each node terminates after completing iteration $t_{end}$, where $t_{end}$ is a constant defined later in Equation (\ref{eq:tend}). The value of $t_{end}$ depends on graph $G(\scriptv, \scripte)$, constants $U$ and $\mu$ defined earlier in Section \ref{sec:iacbl} and parameter $\epsilon$ in $\epsilon$-agreement property.

\hrule

\section{Sufficiency (Correctness of Algorithm 1)}
\label{s:sufficiency}

We will prove that given a graph $G(\scriptv, \scripte)$ satisfying {\em Condition S}, Algorithm 1 is correct, i.e., Algorithm 1 satisfies {\em termination, validity, $\epsilon$-agreement} properties. Therefore, {\em Condition S} and {\em Condition P} are proved to be sufficient. We borrow the matrix analysis from the work on non-fault-tolerant consensus \cite{Jadbabaie,AA_convergence_markov}. The proof below follows the same structure in our prior work on node failures \cite{Tseng_general,vaidya_icdcn14}; however, such analysis has not been applied in the case of link failures.

In the rest of the section, we assume that $G(\scriptv,\scripte)$ satisfies {\em Condition S}  and {\em Condition P}. We introduce standard matrix tools to facilitate our proof. Then, we use transition matrix to represent the {\em Update} step in Algorithm 1, and show how to use these tools to prove the correctness of Algorithm 1 in $G(\scriptv,\scripte)$.

\subsection{Matrix Preliminaries}

In the discussion below, we use boldface upper case letters to denote matrices,
rows of matrices, and their elements. For instance,
$\bfA$ denotes a matrix, $\bfA_i$ denotes the $i$-th row of
matrix $\bfA$, and $\bfA_{ij}$ denotes the element at the
intersection of the $i$-th row and the $j$-th column
of matrix $\bfA$.

\begin{definition}
\label{d_stochastic}
A vector is said to be {\em stochastic} if all the elements
of the vector are non-negative, and the elements add up to 1.
A matrix is said to be {\em row stochastic} if each row of the matrix is a
stochastic vector. 
\end{definition}

When presenting matrix products, for convenience of presentation, we adopt the ``backward'' product convention below, where $a \leq b$,

\begin{equation}
\label{backward}
\Pi_{i=a}^b \bfA[i] = \bfA[b]\bfA[b-1]\cdots\bfA[a]
\end{equation}

For a row stochastic matrix $\bfA$,
 coefficients of ergodicity $\delta(\bfA)$ and $\lambda(\bfA)$ are defined as
follows \cite{Wolfowitz}:
\begin{eqnarray*}
\delta(\bfA) & = &   \max_j ~ \max_{i_1,i_2}~ | \bfA_{i_1\,j}-\bfA_{i_2\,j} | \label{e_zelta} \\
\lambda(\bfA) & = & 1 - \min_{i_1,i_2} \sum_j \min(\bfA_{i_1\,j} ~, \bfA_{i_2\,j}) \label{e_lambda}
\end{eqnarray*}

\begin{lemma}
\label{lemma:ergodicity}
For any $p$ square row stochastic matrices $\bfA(1), \bfA(2), \dots, \bfA(p)$,

\begin{equation*}
\delta(\Pi_{u=1}^p \bfA(u)) \leq \Pi_{u=1}^p \lambda(\bfA(u))
\end{equation*}
\end{lemma}

Lemma \ref{lemma:ergodicity} is proved in \cite{Hajnal58}. Lemma \ref{lemma:ergodicity2} below follows from the definition of $\lambda(\cdot)$.

\begin{lemma}
\label{lemma:ergodicity2}
If all the elements in any one column of matrix \bfA are lower bounded by a constant $\gamma$, then $\lambda(\bfA) \leq 1 - \gamma$. That is, if $\exists g$, such that $\bfA_{ig} \geq \gamma \forall i$, then $\lambda(\bfA) \leq 1 - \gamma$.
\end{lemma}

It is easy to show that  $0\leq \delta(\bfA) \leq 1$ and $0\leq \lambda(\bfA) \leq 1$, and that the rows
of $\bfA$ are all identical iff $\delta(\bfA)=0$. Also, $\lambda(\bfA) = 0$ iff $\delta(\bfA) = 0$.

\subsection{Correctness of Algorithm 1}
Denote by $v[0]$ the column vector consisting of the initial states at all nodes. The $i$-th element of $v[0]$, $v_i[0]$, is the initial state of node $i$. Denote by $v[t]$, for $t \geq 1$, the column vector consisting of the states of all nodes at the end
of the $t$-th iteration. The $i$-th element of vector $v[t]$ is state $v_i[t]$. 

For $t \geq 1$, define $F[t]$ to be the set of all links behaving faulty in iteration $t$. Recall that link $(j,i)$ is said to be faulty in iteration $t$ if the value received by node $i$ is different from what node $j$ sends in iteration $t$. Then, define $N_i^F$ as the set of all nodes whose outgoing links to node $i$ is faulty in iteration $t$, i.e., $N_i^F = \{j~|~j \in N_i^-,~(j,i) \in F[t]\}$.\footnote{$N_i^F$ may be different for each iteration $t$. For simplicity, the notation does not explicitly represent this dependence.}

Define $N_i^r$ as a subset of incoming neighbors at node $i$ of size at most $f$, i.e.,\footnote{As will be seen later, $N_i^r$ corresponds to the links removed in some link-reduced graph. Thus, the superscript $r$ in the notation stands for ``removed." $N_i^r$ may be different for each $t$. For simplicity, the notation does not explicitly represent this dependence.}

\[
N_i^r \subseteq N_i^-~~~~\text{such that}~~|N_i^r| \leq f
\]

Now, we state the key lemma that helps prove the correctness of Algorithm 1. In particular, Lemma \ref{lemma:tm2cm} allows us to use results for non-homogeneous Markov chains to prove the correctness of Algorithm 1.  The proof is presented in Appendix \ref{a:tm2cm}.

\begin{lemma}
\label{lemma:tm2cm}
The {\em Update} step in iteration $t~(t \geq 1)$ of Algorithm 1 at the nodes can be expressed as 

\begin{equation}
\label{matrix:e_T}
v[t] = \matrixm[t] v[t-1]
\end{equation}
where {\normalfont$\matrixm[t]$} is an $n \times n$ row stochastic transition matrix with the following property: there exist a constant $\beta~(0 < \beta \leq 1)$ that depends only on graph $G(\scriptv, \scripte)$, and $N_i^r$ such that for each $i \in \scriptv$, and  for all $j \in \{i\}\cup(N_i^- - N_i^F - N_i^r)$, 

\[
\matrixm_{ij}[t] \geq \beta
\]
\end{lemma}

~

Matrix $\matrixm[t]$ is said to be a \underline{transition matrix} for iteration $t$. Aa the lemmas states, $\matrixm[t]$ is a row stochastic matrix. The proof of Lemma \ref{lemma:tm2cm} shows how to construct a suitable row stochastic matrix $\matrixm[t]$ for each iteration $t$. $\matrixm[t]$ depends not only on $t$ but also on the behavior of the faulty links in iteration $t$. 

\begin{theorem}
\label{thm:correctness}
Algorithm 1 satisfies the {\em Termination, Validity}, and {\em $\epsilon$-agreement} properties.
\end{theorem}

\begin{proof}
Sections \ref{s:validity}, \ref{s:termination} and \ref{s:agreement} provide the proof that Algorithm 1 satisfies the three properties for iterative approximate consensus in the presence of Byzantine links. This proof follows a structure used to prove correctness
of other consensus algorithms in our prior work \cite{Tseng_general,vaidya_icdcn14}.
\fillbox
\end{proof}

\subsection{Validity Property}
\label{s:validity}

Observe that $\matrixm[t+1](\matrixm[t] v[t-1]) = (\matrixm[t+1]\matrixm[t]) v[t-1]$. Therefore, by repeated application of (\ref{matrix:e_T}), we obtain for $t \geq 1$,

\begin{equation}
\label{eq:matrixT}
v[t] = (\Pi_{u = 1}^t \matrixm[u]) v[0]
\end{equation}

Since each $\matrixm[u]$ is row stochastic as shown in Lemma \ref{lemma:tm2cm}, the matrix product $\Pi_{u=1}^t \matrixm[u]$ is also a row stochastic matrix. Thus, (\ref{eq:matrixT}) implies that the state of each node $i$ at the end of iteration $t$ can be expressed as a convex combination of the initial states at all the nodes. Therefore, the validity property is satisfied.

\subsection{Termination Property}
\label{s:termination}

Algorithm 1 terminates after $t_{end}$ iterations, where $t_{end}$ is a finite constant depending only on $G(\scriptv, \scripte), U, \mu$, and $\epsilon$. Recall that $U$ and $\mu$ are defined as upper and lower bounds of the initial inputs at all nodes, respectively. Therefore, trivially, the algorithm satisfies the termination property. Later, using (\ref{eq:tend}), we define a suitable value for $t_{end}$.

\subsection{$\epsilon$-agreement Property}
\label{s:agreement}

The proof below follows the same structure in our prior works on node failures \cite{Tseng_general,vaidya_icdcn14} for proving correctness of other consensus algorithms with Byzantine nodes.

Denote by $R_F$ the set of all the link-reduced graph of $G(\scriptv, \scripte)$ corresponding to some faulty link set $F$. Let

\[
r = \sum_{F \subset \scripte,~|F| \leq f} |R_F|
\]

Note that $r$ only depends on $G(\scriptv, \scripte)$ and $f$, and is a finite integer.

Consider iteration $t~(t\geq 1)$. Recall that $F[t]$ denote the set of faulty links in iteration $t$. Then for each link-reduced graph $\graphh[t] \in R_{F[t]}$, define connectivity matrix $\matrixh[t]$ as follows, where $1 \leq i,j \leq n$:

\begin{itemize}
\item $\matrixh_{ij}[t] = 1$, if either $j = i$, or edge $(j,i)$ exists in link-reduced graph $\graphh$;
\item $\matrixh_{ij}[t] = 0$, otherwise.
\end{itemize}

Thus, the non-zero elements of row $\matrixh_i[t]$ correspond to the incoming links at node $i$ in the link-reduced graph $\graphh[t]$, or the self-loop at $i$. Observe that $\matrixh[t]$ has a non-zero diagonal.

Based on {\em Condition S} and Lemma \ref{lemma:tm2cm}, we can show the following key lemmas.

\begin{lemma}
\label{lemma:non-zero}
For any $\graphh[t] \in R_{F[t]}$, and $k \geq n,~{\normalfont\bf \matrixh}^{k}[t]$ has at least one non-zero column, i.e., a column with all elements non-zero.
\end{lemma}

\begin{proof}
$G(\scriptv,\scripte)$ satisfies the {\em Condition S}. Therefore, by Lemma \ref{lemma:path},
there exists at least one node $p$ in the link-reduced graph $\graphh[t]$ that has directed paths to all the nodes in $\graphh[t]$ (consisting of the edges in $\graphh[t]$). $\matrixh^k_{jp}[t]$ of product $\matrixh^k[t]$ is $1$ if and only if node $p$ has a directed path to node $j$ consisting of at most $k$ edges in $\graphh[t]$. Since the length of the path from $p$ to any other node in $\graphh[t]$ is at most $n$, and $p$ has directed paths to all the nodes, for $k \geq n$ the $p$-th column of matrix $\matrixh^{k}[t]$ will be non-zero.\footnote{That is, all the elements of the column will be non-zero. Also, such a non-zero column will exist in $\matrixh^{n-1}[t]$, too. We use the loose bound of $n$ to simplify the presentation.} 
\fillbox
\end{proof}

For matrices $\textbf{A}$ and $\textbf{B}$ of identical dimension, we say that $\textbf{A} \leq \textbf{B}$ iff $\gamma \textbf{A}_{ij} \leq \textbf{B}_{ij}$ for all $i, j$. Lemma below relates the transition matrices with the connectivity matrices. Constant $\beta$ used in the lemma below was introduced in
Lemma \ref{lemma:tm2cm}.

\begin{lemma}
\label{lemma:cm}
For any $t \geq 1$, there exists a link-reduced graph $\graphh[t] \in R_{F[t]}$ such  that $\beta {\normalfont\bf \matrixh[t] \leq \matrixm}[t]$, where $\matrixh[t]$ is the connectivity matrix for $\graphh[t]$.
\end{lemma}

\begin{proof}
First, let us construct a link-reduced graph $\graphh[t]$ by first removing $F[t]$ from $G(\scriptv, \scripte)$. Recall that $F[t]$ is the set of faulty links in iteration $t$. Then for each $i$, remove a set of at most $f$ node $i$'s incoming links as defined in Lemma \ref{lemma:tm2cm} ($N_i^r$). As a result, we have obtained a link-reduced graph $\graphh[t]$ such that $\matrixm_{ij}[t] \geq \beta$, if $j = i$ or edge
$(j, i)$ is in the link-reduced graph $\graphh[t]$.

Denote by $\matrixh[t]$ the connectivity matrix for the link-reduced graph $\graphh[t]$. Then, $\matrixh_{ij}[t]$ denotes the element
in $i$-th row and $j$-th column of $\matrixh[t]$. By definition of the connectivity matrix, we know that $\matrixh_{ij}[t] = 1$, if $j = i$ or edge $(j, i)$ is in the link-reduced graph; otherwise, $\matrixh_{ij}[t] = 0$.

The statement in the lemma then follows from the above two observations.
\fillbox
\end{proof}

\begin{lemma}
\label{l_product_H}
For any $z\geq 1$,
at least one column in the matrix product $\Pi_{t=u}^{u+r n-1} \, \bfH[t]$ is non-zero. 
\end{lemma}

\begin{proof}
Since $\Pi_{t=u}^{u+r n-1} \, \bfH[t]$ consists of $r n$ connectivity matrices
corresponding to link-reduced graphs, and the number of all link-reduced graphs for $F$ ($|F| \leq f$) is $r$,
connectivity matrices corresponding to at least one link-reduced graph, say matrix $\bfH_*$\,, will appear in the above product at least $n$ times.

Now observe that: (i)
By Lemma \ref{lemma:non-zero}, $\bfH_*^{n}$ contains a non-zero
column, say the $k$-th column is non-zero,
and (ii) by definition, all the $\bfH[t]$ matrices in the product contain a non-zero diagonal. These two observations together imply that the $k$-th column in the above product is non-zero.\footnote{The product $\Pi_{t=u}^{u+r n-1} \, \bfH[t]$ can be viewed as the product of $n$ instances of $\matrixh_*$ ``interspersed" with matrices with non-zero diagonals.}
\fillbox
\end{proof}

Let us now define a sequence of matrices $\bfQ(i)$, $i\geq 1$, such that
each of these matrices is a product of $r n$ of the
$\bfM[t]$ matrices. Specifically,
\begin{eqnarray}
\bfQ(i) &=& \Pi_{t=(i-1)r n+1}^{i r n} ~ \bfM[t]
\label{e_Q_i}
\end{eqnarray}
From (\ref{eq:matrixT}) and (\ref{e_Q_i})
observe that
\begin{eqnarray}
\bfv[k r n] & = & \left(\, \Pi_{i=1}^k ~ \bfQ(i) \,\right)~\bfv[0]
\end{eqnarray}

\begin{lemma}
\label{l_Q}
For $i\geq 1$, $\bfQ(i)$ is a scrambling row stochastic matrix,
and \[ \lambda(\bfQ(i))\leq 1-\beta^{r n}.\]
\end{lemma}

\begin{proof}

$\bfQ(i)$ is a product of row stochastic matrices ($\bfM[t]$); therefore,
$\bfQ(i)$ is row stochastic.
From Lemma \ref{lemma:cm}, for each $t\geq 1$,
\[
\beta \, \bfH[t] ~ \leq ~ \bfM[t]
\]
Therefore, 
\[
\beta^{r n} ~ \Pi_{t=(i-1)r n+1}^{i r n} ~ \bfH[t] ~ \leq 
~ \Pi_{t=(i-1) r n+1}^{i r n} ~ \bfM[t] ~ =
~ \bfQ(i)
\]
By using $u=(i-1)n+1$ in Lemma \ref{l_product_H},
we conclude that the matrix product on the left side
of the above inequality contains a non-zero column. Therefore, since $\beta > 0$, $\bfQ(i)$ on the
right side of the inequality also contains
a non-zero column.

Observe that $r n$ is finite, and hence, $\beta^{r n}$
is non-zero. Since the non-zero terms in $\bfH[t]$ matrices are all 1,
the non-zero elements in $\Pi_{t=(i-1) r n+1}^{i r n} \bfH[t]$
must each be $\geq$ 1. Therefore, there exists a non-zero column in $\bfQ(i)$
with all the elements in the column being $\geq \beta^{r n}$.
Therefore, by Lemma \ref{lemma:ergodicity2}, $\lambda(\bfQ(i))\leq 1-\beta^{r n}$, and $\bfQ(i)$ is a scrambling matrix. 
\fillbox
\end{proof}

Let us now continue with the proof of $\epsilon$-agreement. Consider the coefficient of ergodicity $\delta(\Pi_{u=1}^t \matrixm[u])$.

\begin{align}
\delta(\Pi_{u=1}^t \bfM[u]) &= \delta\left(
\left(\Pi_{u=(\lfloor\frac{t}{r n}\rfloor)r n+1}^t \bfM[u]\right)
\left(\Pi_{u=1}^{\lfloor\frac{t}{r n}\rfloor} \bfQ(i)\right)\right)~~~\text{by definition of}~~\bfQ(u)\nonumber\\
&\leq \lambda
\left(\Pi_{u=(\lfloor\frac{t}{r n}\rfloor)r n+1}^t \bfM[u]\right)
\left(\Pi_{u=1}^{\lfloor\frac{t}{r n}\rfloor} \lambda\left( \bfQ(u)\right)\right)~~~\text{by Lemma \ref{lemma:ergodicity}}\nonumber\\
&\leq \Pi_{u=1}^{\lfloor\frac{t}{r n}\rfloor} \lambda\left(\bfQ(u)\right)~~~\text{because}~~\lambda(\cdot) \leq 1\nonumber\\
&\leq \left( 1-\beta^{rn}\right)^{\lfloor\frac{t}{r n}\rfloor}~~~\text{by Lemma \ref{l_Q}}\label{eq:t}
\end{align}

Observe that the upper bound on right side of (\ref{eq:t}) depends only on graph $G(\scriptv, \scripte)$ and $t$, and is
independent of the input states, and the behavior of the faulty links. Moreover, the
upper bound on the right side of (\ref{eq:t}) is a non-increasing function of $t$. Define $t_{end}$ as the smallest positive integer such that the right hand side of (\ref{eq:t}) is smaller than $\frac{\epsilon}{n \max(|U|, |\mu|)}$. Recall that $U$ and $\mu$ are defined as the upper and lower bound of the inputs at all nodes. Thus,

\begin{equation}
\label{eq:tend}
\delta(\Pi_{u=1}^{t_{end}} \bfM[u]) \leq \left( 1-\beta^{rn}\right)^{\lfloor\frac{t_{end}}{r n}\rfloor} <
\frac{\epsilon}{n \max(|U|, |\mu|)}
\end{equation}

Recall that $\beta$ and $r$ depend only on $G(\scriptv, \scripte)$. Thus, $t_{end}$ depends only on graph $G(\scriptv, \scripte)$, and constants $U, \mu$ and $\epsilon$.

Recall that $\Pi_{u=1}^t \matrixm[u]$ is an $n \times n$ row stochastic matrix. let $\matrixm^* = \Pi_{u=1}^t \matrixm[u]$. From \ref{eq:matrixT}, we have $v_j[t] = \matrixm^*_j v[0]$. That is, the state of any node $j$ can be obtained as the product of the $j$-th row of $\matrixm^*$ and $v[0]$. Now, consider any two nodes $j, k$, we have

\begin{align}
|\bfv_j[t] - v_k[t]| &= |\matrixm^*_j v[0] - \matrixm^*_k v[0]| \nonumber\\
&= |\Sigma_{i=1}^n \matrixm^*_{ji} v_i[0] - \Sigma_{i=1}^n \matrixm^*_{ki} v_i[0]| \nonumber\\
&= |\Sigma_{i=1}^n \left(\matrixm^*_{ji} - \matrixm^*_{ki} \right) v_i[0]| \nonumber\\
&\leq \Sigma_{i=1}^n | \matrixm^*_{ji} - \matrixm^*_{ki} | |v_i[0]| \nonumber\\
&\leq \Sigma_{i=1}^n \delta(\matrixm^*) |v_i[0]| \nonumber\\
&\leq n \delta(\matrixm^*) \max(|U|, |\mu|) \nonumber\\
&\leq n \delta(\Pi_{u=1}^t \matrixm[u]) \max(|U|, |\mu|) \label{eq:delta}
\end{align}

Therefore, by (\ref{eq:tend}) and (\ref{eq:delta}), we have

\begin{equation}
\label{eq:final}
|v_j[t_{end}] - v_k[t_{end}]| < \epsilon
\end{equation}

Since the output of the nodes equal its state at termination (after $t_{end}$ iterations). Thus, (\ref{eq:final}) implies that Algorithm 1 satisfies the $\epsilon$-agreement property.

\section{Summary}

This paper explores approximate consensus problem under transient Byzantine link failure model. We address a particular class of iterative algorithms in arbitrary directed graphs, and prove a necessary and sufficient condition for the graphs to be able to solve the approximate consensus problem iteratively.



\begin{thebibliography}{10}
\bibitem{AA_optimal}
I. Abraham,~Y. Amit,~and~D. Dolev. 
\newblock Optimal resilience asynchronous approximate agreement. 
\newblock In OPODIS, 2004.



\bibitem{Biely_hybrid}
M.~Biely, U.~Schmid, and B.~Weiss. 
\newblock {\em Synchronous consensus under hybrid process and link failures}. \newblock Theoretical Computer Science, 412(40):5602 – 5630, 2011.

\bibitem{AA_convergence_markov}
D.~P. Bertsekas and J.~N. Tsitsiklis.
\newblock {\em Parallel and Distributed Computation: Numerical Methods}.
\newblock Optimization and Neural Computation Series. Athena Scientific, 1997.

\bibitem{HeardOf}
B. Charron-Bost and A. Schiper. 
\newblock The Heard-Of model: computing in distributed systems with benign faults.
\newblock Distributed Computing, 22(1):49–71, April 2009.

\bibitem{dag_decomposition}
S.~Dasgupta, C.~Papadimitriou, and U.~Vazirani.
\newblock {\em Algorithms}.
\newblock McGraw-Hill Higher Education, 2006.

\bibitem{AA_Dolev_1986}
D.~Dolev, N.~A. Lynch, S.~S. Pinter, E.~W. Stark, and W.~E. Weihl.
\newblock Reaching Approximate Agreement in the presence of Faults.
\newblock {\em J. ACM}, May 1986.

\bibitem{impossible}
M. J. Fischer, N. A. Lynch, and M. Merritt. 
\newblock Easy impossibility proofs for distributed consensus problems.
\newblock PODC '85, 1985. ACM.

\bibitem{Hajnal58}
J.~Hajnal.
\newblock Weak Ergodicity in non-homogeneous Markov Chains.
\newblock In {\em Proceedings of the Cambridge Philosophical Society},
  volume~54, pages 233--246, 1958.

\bibitem{Jadbabaie}
A.~Jadbabaie, J. Lin, and A. Morse. 
\newblock Coordination of Groups of Mobile Autonomous Agents using Nearest Neighbor Rules. 
\newblock Automatic Control, IEEE Transactions on, 48(6):988--1001, June 2003.


\bibitem{Kempe_gossip}
D.~Kempe,~A.~Dobra,~and J. Gehrke.
\newblock Gossip-based computation of aggregate information.
\newblock IEEE Symposium on Foundations of Computer Science, Oct. 2003.

\bibitem{psl_BG_1982}
L.~Lamport, R.~Shostak, and M.~Pease.
\newblock The Byzantine Generals Problem.
\newblock {\em ACM Trans. on Programming Languages and Systems}, 1982.

\bibitem{Sundaram_condition}
H. J. LeBlanc, H. Zhang, X. Koutsoukos, S. Sundaram.
\newblock Resilient Asymptotic Consensus in Robust Networks.
\newblock Selected Areas in Communications, IEEE Journal on , vol.31, no.4, pp.766,781, April 2013.

\bibitem{packet}
D. S. Lun, M. M{\'e}dard, R. Koetter, and M. Effros.
\newblock On coding for reliable communication over packet networks.
\newblock Physical Communication, 2008.

\bibitem{control}
M. Pajic, S. Sundaram, J. Le Ny, G. J. Pappas, and R. Mangharam.
\newblock Closing the Loop: A Simple Distributed Method for Control
over Wireless Networks.
\newblock international conference on Information Processing in Sensor Networks, 2012.


\bibitem{Santoro_link}
N. Santoro, and P. Widmayer. 
\newblock Time is not a healer.
\newblock in: Proc. 6th Ann. Symposium on Theoretical Aspects of Computer Science, STACS '89, 1989.

\bibitem{Santoro_link2}
N. Santoro and P. Widmayer.
\newblock Agreement in synchronous networks with ubiquitous faults.
\newblock Theor. Comput. Sci. 384 (2-3) (2007) 232–249.

\bibitem{noisy_link}
I. D. Schizas, A. Ribeiro, and G. B. Giannakis. 
\newblock Consensus in Ad Hoc WSNs With Noisy Links- Part I: Distributed Estimation of Deterministic Signals.
\newblock IEEE Transactions on Signal Processing, 2008.

\bibitem{Schmid_link}
U. Schmid, B. Weiss, I. Keidar.
\newblock Impossibility results and lower bounds for consensus under link failures.
\newblock SIAM Journal on Computing 38 (5) 1912–1951, 2009..

\bibitem{SS_link}
S. Sundaram, S. Revzen, and G. Pappas.
\newblock A control-theoretic approach to disseminating values and overcoming malicious links in wireless networks
\newblock Automatica, 2012.

\bibitem{SS_node}
S. Sundaram and C. N. Hadjicostis.
\newblock Distributed function calculation via linear iterative strategies in the presence of malicious agent.
\newblock IEEE Transactions on Automatic Control, 2011.

\bibitem{Tseng_general}
L. Tseng and N. H. Vaidya. 
\newblock Iterative approximate byzantine consensus under a generalized fault model.
\newblock In International Conference on Distributed Computing and Networking (ICDCN), January 2013.

\bibitem{vaidya_PODC12}
N.~H. Vaidya, L.~Tseng, and G.~Liang.
\newblock Iterative Approximate Byzantine Consensus in Arbitrary Directed
  Graphs.
\newblock PODC '12, 2012. ACM.

\bibitem{vaidya_icdcn14}
N.~H. Vaidya.
\newblock Iterative Byzantine Vector Consensus in Incomplete Graphs.
\newblock In International Conference on Distributed Computing and Networking (ICDCN), January 2014.


\bibitem{Wolfowitz}
J.~Wolfowitz.
\newblock Products of Indecomposable, Aperiodic, Stochastic Matrices.
\newblock In {\em Proceedings of the American Mathematical Society}, volume~14,
  pages 733--737, 1963.
\end{thebibliography}

\appendix

\section*{Appendix}

\section{Example Network}
\label{a:example}

\begin{lemma}
The graph in Figure \ref{f:eg} satisfies {\em Condition P} when $f = 1$.
\end{lemma}

\begin{proof}
Denote by $G$ the graph in Figure \ref{f:eg}. First observe that a clique of $4$ nodes satisfies {\em Condition P} when $f=1$. Thus, for $G$, we only need to consider the case when node $E$ is in either $L$ or $R$; otherwise, some node in $L$ (or $R$) from the clique (formed by nodes $A, B, C, D$) will have at least $f+1 = 2$ incoming links from $R$ (or $L$) excluding link in $F$. 

Without loss of generality, consider the case when $E$ is in $L$. Consider the following cases:

\begin{itemize}
\item One of the nodes $A, B, C, D$ is in $L$: say node $X$ is in $L$ besides $E$. Then node $X$ has at least $f+1$ incoming links from $R$ excluding link in $F$.

\item Two of the nodes $A, B, C, D$ are in $L$: say nodes $X_1, X_2$ are in $L$ besides $E$. Then either node $X_1$ or $X_2$ has at least $f+1$ incoming links from $R$ excluding link in $F$.

\item Three of the nodes $A, B, C, D$ are in $L$: say node $Y$ is the only node in $R$, since all the other nodes are in $L$. Then node $Y$ has at least $f+1$ incoming links from $L$ excluding link in $F$.

\end{itemize}
In every case, either $L \cup C \rightarrow R$ or $C \cup R \rightarrow L$. Thus, $G$ satisfies {\em Condition P}.

\fillbox
\end{proof}

\section{Proof of Lemma \ref{lemma:tm2cm}}
\label{a:tm2cm}

We prove the following Lemma in Section \ref{s:sufficiency}.

~

\noindent {\bf Lemma} \ref{lemma:tm2cm}.
{\em 
The {\em Update} step in iteration $t~(t \geq 1)$ of Algorithm 1 at the nodes can be expressed as 

\begin{equation}
v[t] = \matrixm[t] v[t-1]
\end{equation}
where {\normalfont$\matrixm[t]$} is an $n \times n$ row stochastic transition matrix with the following property: there exist a constant $\beta~(0 < \beta \leq 1)$ that depends only on graph $G(\scriptv, \scripte)$, and $N_i^r$ such that for each $i \in \scriptv$, and  for all $j \in \{i\}\cup(N_i^- - N_i^F - N_i^r)$, 

\[
\matrixm_{ij}[t] \geq \beta
\]
}

\begin{proof}

We prove the correctness of Lemma \ref{lemma:tm2cm} by constructing $\matrixm_i[t]$ for $1 \leq i \leq n$ that satisfies the conditions in Lemma \ref{lemma:tm2cm}. Recall that $F[t]$ denotes the set of faulty links in the $t$-th iteration.

Consider a node $i$ in iteration $t~(t \geq 1)$. 
In the {\em Update} step of Algorithm 1, recall that the smallest and the largest $f$ values are removed from $r_i[t]$ by node $i$. Denote by $\scripts$ and $\scriptl$, respectively, the set of nodes\footnote{Although $\scripts$ and $\scriptl$ may be different for each iteration $t$, for simplicity, we do not explicitly represent this dependence on $t$ in the notations $\scripts$ and $\scriptl$.} from whom the smallest and the largest $f$ values were received by node $i$ in iteration $t$. Define sets $\scripts_g$ and $\scriptl_g$ to be subsets of $\scripts$ and $\scriptl$ that contain all the nodes from whom node $i$ receives the correct value in $\scripts$ and $\scriptl$, respectively. That is, $\scripts_g = \{j~|~j \in \scripts,~(j,i) \in \scripte - F[t]\}$ and $\scriptl_g = \{j~|~j \in \scriptl,~(j,i) \in \scripte - F[t]\}$.

Construction of $\matrixm_i[t]$ differs somewhat depending on whether
sets $\scripts_g, \scriptl_g$ and $N_i^F$ are empty or not.
We divide the possibilities into 3 separate cases:

\begin{itemize}
\item Case I: $\scripts_g \neq \emptyset, \scriptl_g \neq \emptyset$, and $N_i^F \neq \emptyset$.

\item Case II: $\scripts_g \neq \emptyset, \scriptl_g \neq \emptyset$, and $N_i^F = \emptyset$.

\item Case III: at most one of $\scripts_g$ and $\scriptl_g$, and $N_i^F = \emptyset$.

\end{itemize}
Observe that if $\scripts_g$ ($\scriptl_g$) is empty, then $N_i^F = \emptyset$ and $\scriptl = \scriptl_g$ ($\scripts = \scripts_g$), since there are at most $f$ faulty links and $|\scripts| = |\scriptl| = f$. Therefore, the 3 cases above cover all the possible scenarios.

\noindent {\bf Case I}

~
In Case I,
$\scripts_g \neq \emptyset, \scriptl_g \neq \emptyset$, and $N_i^F \neq \emptyset$. Let $m_{\scripts}$ and $m_{\scriptl}$ be defined as shown below. Recall that the incoming links from the nodes in $\scripts_g$ and $\scriptl_g$ to node $i$ are all fault-free, and therefore, for any node $j\in \scripts_g \cup \scriptl_g$, $w_j=v_j[t-1]$ (in the notation of Algorithm 1). That is, the value received by node $i$ from node $j$ is exactly the state at node $j$ in iteration $t-1$.

\begin{equation*}
m_{\scripts} = \frac{\sum_{j \in \scripts_g} v_j[t-1]}{|\scripts_g|}~~~~~\text{and}~~~~~m_{\scriptl} = \frac{\sum_{j \in \scriptl_g} v_j[t-1]}{|\scriptl_g|}
\end{equation*}
Now, consider any node $k \in N_i^F$. By the definition of sets $\scripts_g$ and $\scriptl_g$,
$m_{\scripts} \leq w_k \leq m_{\scriptl}$. Therefore, we can find weights $S_k \geq 0$ and $L_k \geq 0$ such that $S_k + L_k = 1$, and

\begin{eqnarray}
w_k & = & S_k~m_{\scripts} + L_k~m_{\scriptl} \\
& = & 
\frac{S_k}{|\scripts_g|}
\sum_{j \in \scripts_g} v_j[t-1]
+
\frac{L_k}{|\scriptl_g|}
\sum_{j \in \scriptl_g} v_j[t-1]
\label{eq:caseI}
\end{eqnarray}
Clearly, at least one of $S_k$ and $L_k$ must be $\geq 1/2$.

~
%
%

We now define elements $\matrixm_{ij}[t]$ of row $\matrixm_i[t]$:

\begin{itemize}
\item For $j \in N_i^*[t] - N_i^F$ : In this case, either the edge $(j,i)$ is fault-free, or $j = i$.
For each such $j$, define $\matrixm_{ij}[t] = a_i$. This is obtained by observing
in (\ref{e_T}) that the contribution of such a node $j$ to the new state
$v_i[t]$ is $a_i~w_j = a_i~v_j[t-1]$.

The elements of $\matrixm_i[t]$ defined here add up to $$|N_i^*[t] - N_i^F|~a_i$$

\item For $j\in \scripts_g\cup\scriptl_g$ : In this case, the edge $(j,i)$ is a fault-free.

For each $j \in \scripts_g$,
\[
\matrixm_{ij}[t] ~=~ a_i \, \sum_{k \in N_i^F} \frac{S_k}{|\scripts_g|}
\]
and for each node $j \in \scriptl_g$,
\[
\matrixm_{ij}[t] ~=~ a_i \, \sum_{k \in N_i^F} \frac{L_k}{|\scriptl_g|}
\]
To obtain these two expressions, we represent value $w_k$ sent via faulty link $(k,i)$ for each $k \in N_i^F$ using (\ref{eq:caseI}).  
Recall that this node $k$ contributes $a_iw_k$ to (\ref{e_T}).
The above two expressions are then obtained by summing (\ref{eq:caseI})
over all the nodes in $N_i^F$, and replacing this sum
by equivalent contributions by nodes in $\scripts_g$ and $\scriptl_g$.

The elements of $\matrixm_i[t]$ defined here add up to $a_i \, \sum_{k \in N_i^F} (S_k + L_k) = |N_i^F|~a_i$

\item For $j\in \scriptv - ((N_i^* - N_i^F) \cup \scripts_g \cup \scriptl_g)$ :
These nodes have not yet been considered above.
For each such node $j$, define $\matrixm_{ij}[t] = 0$.
\end{itemize}
With the above definition of $\matrixm_i[t]$, it should be easy to see
that $\matrixm_i[t]\,v[t-1]$ is, in fact, identical to $v_i[t]$ obtained using (\ref{e_T}). Thus, the above construction of $\matrixm_i[t]$ results in the values sent via faulty links to (\ref{e_T}) being replaced by an equivalent contribution from the nodes in $\scriptl_g$ and $\scripts_g$.

~

\paragraph{Properties of $\matrixm_i[t]$:}

First, we show that $\matrixm[t]$ is row stochastic. Observe that all
the elements of $\matrixm_i[t]$ are non-negative.
Also, all the elements of $\matrixm_i[t]$ above add up to
\[
|N_i^*[t] - N_i^F|~a_i + |N_i^F|~a_i = |N_i^*[t]|~a_i = 1
\]
because $a_i = 1/|N_i^*[t]|$ as defined in Algorithm 1.
Thus, $\matrixm_i[t]$ is a stochastic row vector.

Recall that from the above discussion, for $k\in N_i^F$,
one of $S_k$ and $L_k$ must be $\geq 1/2$.
Without loss of generality, assume that $S_s \geq 1/2$ for all nodes $s \in N_i^F$.
Consequently, for each node $j \in \scripts_g$, $\matrixm_{ij}[t] \geq \frac{a_i}{|\scripts_g|} S_s \geq \frac{a_i}{2|\scripts_g|}$. Also,
for each node $j$ in $N_i^*[t] - N_i^F$,
$\matrixm_{ij}[t] = a_i$.
Thus, if $\beta$ is chosen such that
\begin{equation}
\label{eq:beta_caseI}
0 < \beta \leq \frac{a_i}{2|\scripts_g|}
\end{equation}
and $N_i^r$ is defined to be $\scriptl_g$, then the condition in the
lemma holds for node $i$. That is, 
for all $j \in \{i\}\cup(N_i^- - N_i^F - N_i^r)$, 

\[
\matrixm_{ij}[t] \geq \beta
\]

~

\noindent {\bf Case II}

~

Now, we consider the case when $\scripts_g \neq \emptyset, \scriptl_g \neq \emptyset$, and $N_i^F = \emptyset$. That is, when each of $\scripts$ and $\scriptl$ contains at least one node from which the node $i$ receives correct value, and node $i$ receives correct value(s) from all the node(s) in $N_i^*[t]$. In fact, the analysis of Case II is very similar to the analysis presented above in Case I. We now discuss how the analysis of Case I can be applied to Case II. Rewrite (\ref{e_T}) as follows:

\begin{eqnarray}
v_i[t] & = & \frac{a_i}{2} v_i[t-1]  + \frac{a_i}{2} v_i[t-1]
	+ \sum_{j\in N_i^*[t] - \{i\}} a_iw_j \\
& = & a_iw_z + a_i w_i
	+ \sum_{j\in N_i^*[t] - \{i\}} a_iw_j 
\end{eqnarray}

In the above equation, $z$ is to be viewed as a ``virtual'' incoming
neighbor of node $i$, which has sent value $w_z=\frac{v_i[t-1]}{2}$
to node $i$ in iteration $t$.
With the above rewriting of state update,
the value received by node $i$ from itself should
be viewed as $w_i=\frac{v_i[t-1]}{2}$ instead of $v_i[t-1]$. 
With this transformation, Case II now becomes identical
to Case I, with virtual node $z$ being treated
as an incoming neighbor of node $i$.

In essence, a part of node $i$'s contribution (half, to be precise) is now replaced by equivalent contribution by nodes in $\scriptl_g$ and $\scripts_g$. 
We now define elements $\matrixm_{ij}[t]$ of row $\matrixm_i[t]$:

\begin{itemize}
\item For $j = i$: $\matrixm_{ij}[t] = \frac{a_i}{2}$. This is obtained by observing in (\ref{e_T}) that node $i$'s contribution to the new state $v_i[t]$ is $a_i\frac{v_i[t-1]}{2}$.

\item For $j \in N_i^*[t] - \{i\}$ : In this case, $j$ is a node from which node $i$ receives correct value. For each such $j$, define $\matrixm_{ij}[t] = a_i$. This is obtained by observing in (\ref{e_T}) that the contribution of node $j$ to the new state $v_i[t]$ is $a_iw_j = a_i v_j[t - 1]$.

\item For $j\in \scripts_g \cup \scriptl_g$ : In this case, $j$ is a node in $\scripts$ or  $\scriptl$ from which node $i$ receives correct value. 

For each $j \in \scripts_g$,
\[
\matrixm_{ij}[t] ~=~ \frac{a_i}{2} \, \frac{S_z}{|\scripts_g|}
\]
and for each node $j \in \scriptl_g$,
\[
\matrixm_{ij}[t] ~=~ \frac{a_i}{2} \, \frac{L_z}{|\scriptl_g|}
\]

where $S_z$ and $L_z$ are chosen such that $S_z + L_z = 1$ and $w_z = \frac{v_i[t-1]}{2} = \frac{S_z}{2} m_{\scripts} + \frac{L_z}{2} m_{\scriptl}$. Note that such $S_z$ and $L_z$ exist because by definition of $\scripts_g$ and $\scriptl_g$, $v_i[t-1] \geq w_j,~\forall j \in S_g$ and $v_i[t-1] \leq w_j,~\forall j \in L_g$. Then the two expressions above are obtained by replacing the contribution of the virtual node $z$ by an equivalent contribution by the nodes in $\scripts_g$ and $\scriptl_g$, respectively.

\item For $j\in \scriptv - (N_i^*[t] \cup \scripts_g \cup \scriptl_g)$ :
 These nodes have not yet been considered above.
For each such node $j$, define $\matrixm_{ij}[t] = 0$.
\end{itemize}

~

\paragraph{Properties of $\matrixm_i[t]$:}

By argument similar to that in {\em Case I}, $\matrixm_i[t]$ is row stochastic. Without loss of generality, suppose that $S_z \geq 1/2$. Then for each node $j \in \scripts_g$, $\matrixm_{ij}[t] = \frac{a_i}{2|\scripts_g|}S_z \geq \frac{a_i}{4|\scripts_g|}$. Also, for node $j$ in $N_i^*[t]-\{i\}$, $\matrixm_{ij}[t] = a_i$, and $\matrixm_{ii}[t] = \frac{a_i}{2}$. Recall that by definition, $|\scripts_g| \geq 1$. Hence, if $\beta$ is chosen such that 

\begin{equation}
\label{eq:beta-caseII}
0 < \beta \leq \frac{a_i}{4|\scripts_g|}
\end{equation}
and $N_i^r$ is defined to be equal to $\scriptl_g$, then the condition in the Lemma \ref{lemma:tm2cm} holds for node $i$. That is,
$\matrixm_{ij}[t] \geq \beta$ for $j \in \{i\} \cup (N_i^- - N_i^F - N_i^r)$. 

~

\noindent {\bf Case III}

Here, we consider the case when at most one of $\scripts_g$ and $\scriptl_g$ is empty, and $N_i^F = \emptyset$. Without loss of generality, suppose that $\scripts$ contains only nodes whose outgoing links to node $i$ is faulty in iteration $t$, i.e., $\scripts = \{j~|~(j,i) \in F[t]\}$. Since there are at most $f$ faulty links and $|\scripts| = f$, $\scriptl = \scriptl_g$. That is, the value received from each node in $\scriptl$ by node $i$ is correct.

In this case, define $\matrixm_{ij}[t]=a_i$ for $j\in N_i^*[t]$;
define $\matrixm_{ij}=0$ for all other nodes $j$.

~

\noindent {\em Properties of $\matrixm_i[t]$:}

All the elements of $\matrixm_i[t]$ are non-negative.
The elements of $\matrixm_i[t]$ defined above add up to
\[
|N_i^*[t]|~a_i = 1
\]
Thus, $\matrixm_i[t]$ is a stochastic row vector.

In Case III, recall that for any node $j$ in $N_i^*[t]$, $\matrixm_{ij}[t] = a_i$.
Thus, if $\beta$ is chosen such that
\begin{equation}
\label{eq:beta-caseIII}
0 < \beta \leq a_i
\end{equation}
and $N_i^r$ is defined to be equal to $\scriptl$,
then the condition in the
Lemma \ref{lemma:tm2cm} holds for node $i$. 

~

\noindent {\bf Putting Cases Together}

Now, let us consider Cases I-III together. From the definition of $a_i$ in Algorithm 1, observe that $a_i \geq \frac{1}{|N_i^-|+1}$ (because $f\geq 0$). Let us define
\[ \alpha = \min_{i\in\scriptv} \frac{1}{|N_i^-|+1}\]
Moreover, observe that $|\scripts_g| \leq n$ and $|\scriptl_g| \leq n$. Then define $\beta$ as 
\begin{equation}
\label{eq:beta}
\beta = \frac{\alpha}{4n}
\end{equation}
This definition satisfies constraints on $\beta$ in Cases I through III (conditions (\ref{eq:beta_caseI}), (\ref{eq:beta-caseII}) and (\ref{eq:beta-caseIII})). Thus, Lemma \ref{lemma:tm2cm} holds for all three cases with this choice in (\ref{eq:beta}).

\fillbox
\end{proof}

\end{document}